\documentclass[a4paper]{article}
\usepackage{graphicx}
\usepackage{onecolceurws}
\usepackage{amsthm, amsmath}
\usepackage{xspace}
\usepackage[utf8x]{inputenc}
\usepackage[LGR, T1]{fontenc}
\PrerenderUnicode{ä}
\PrerenderUnicode{é}

\usepackage{textcomp}
\usepackage{textgreek}
\usepackage{upgreek}
\usepackage{mathrsfs}

\usepackage{listings}

\usepackage{xcolor}
\definecolor{keywordcolor}{rgb}{0.7, 0.1, 0.1}   
\definecolor{commentcolor}{rgb}{0.4, 0.4, 0.4}   
\definecolor{symbolcolor}{rgb}{0.0, 0.1, 0.6}    
\definecolor{sortcolor}{rgb}{0.1, 0.5, 0.1}      

\usepackage{amssymb}
\usepackage{hyperref}

\title{Formalizing the Gromov-Hausdorff space}

\newcommand{\N}{\mathbb{N}}
\newcommand{\R}{\mathbb{R}}

\newcommand{\boK}{\mathcal{K}}
\newcommand{\GH}{\mathcal{GH}}
\newcommand{\mathlib}{\texttt{mathlib}\xspace}

\newtheorem{thm}{Theorem}[section]
\newtheorem{prop}[thm]{Proposition}
\newtheorem{definition}[thm]{Definition}

\theoremstyle{definition}

\author{
Sébastien Gouëzel
\\
IRMAR, CNRS UMR 6625,
Université de Rennes 1, 35042 Rennes, France
\\ sebastien.gouezel@univ-rennes1.fr}

\institution{}

\begin{document}
\maketitle

\begin{abstract}
The Gromov-Hausdorff space is usually defined in textbooks as ``the space
of all compact metric spaces up to isometry''. We describe a formalization
of this notion in the Lean proof assistant, insisting on how we need to
depart from the usual informal viewpoint of mathematicians on this object
to get a rigorous formalization.
\end{abstract}
\vskip 20pt


The Gromov-Hausdorff space is the space of all nonempty compact metric spaces
up to isometry. It has been introduced by Gromov in~\cite{gromov_hausdorff},
and plays now an important role in branches of geometry and probability
theory. Its intricate nature of a space of equivalence classes of spaces
gives rise to interesting formalization questions, both from the point of
view of the interface with the rest of the library and on design choices for
definitions and proofs. This text is devoted to a discussion of these issues:
it describes a formalization of the main features of the Gromov-Hausdorff
space in the Lean proof assistant, developed at Microsoft Research by
Leonardo de Moura~\cite{demoura_lean}, within the library
\mathlib~\cite{mathlib}.

This text is written with two audiences in mind: it can be read by curious
mathematicians who want to learn the basics of the Gromov-Hausdorff space,
and by formalizers who want to learn about the challenges raised by the
formalization of an unusual mathematical object such as this one. It should
be reasonably self-contained.

In Section~\ref{sec:math}, we give a purely mathematical description of the
Gromov-Hausdorff space and its salient features. In Section~\ref{sec:formal},
we give an overview of our formalization. The last three sections are devoted
to specific interesting points that were raised during this formalization.
More specifically, Section~\ref{sec:logic} discusses the possible choices of
definition for the Gromov-Hausdorff space. Section~\ref{sec:realized}
explains how preexisting gaps in the \mathlib library had to be filled to
show that the Gromov-Hausdorff distance is realized.
Section~\ref{sec:complete} focuses on a particularly subtle inductive
construction involved in the proof of the completeness of the
Gromov-Hausdorff space, and the shortcomings of Lean 3 that had to be
circumvented to formalize it.

\section{A primer on the Gromov-Hausdorff space}
\label{sec:math}

In this paragraph, we give a quick overview on the Gromov-Hausdorff space as
presented in mathematics textbooks. See for
instance~\cite[Section~7.3]{burago_burago_ivanov}
or~\cite[Section~10.1]{petersen_book}.

Given two nonempty bounded subsets $A$ and $B$ of a metric space $X$, there
is a way to tell how close these are, as subsets of $X$, through their
\emph{Hausdorff distance} $d^X_H(A, B)$. It is the infimum of those $r$ such
that $A$ is included in the $r$-neighborhood of $B$ (i.e., the set of points
within distance at most $r$ of a point in $B$) and $B$ is included in the
$r$-neighborhood of $A$. This $r$ is finite as $A$ and $B$ are bounded and
nonempty, and zero if and only if $A$ and $B$ have the same closure. In
particular, $d^X_H$ induces a distance on the space of nonempty bounded
closed subsets of $X$, and also on the space $\boK_X$ of its nonempty compact
subsets. This distance is very well behaved: $(\boK_X, d^X_H)$ is complete
(resp.\ second-countable, resp.\ compact) if $X$ is.

Much more recently, Gromov has introduced in~\cite{gromov_hausdorff} a way to
compare metric spaces even when they are not embedded in a common space. His
motivation was to be able to prove that some classes of Riemannian manifolds
were totally bounded or compact, in a suitable sense, to deduce uniformity
statements over all manifolds in these classes. While there are many variants
of his notion of distance, we will focus in this article on the simplest one,
over nonempty compact metric spaces.

\begin{definition}
Let $X$ and $Y$ be two nonempty compact metric spaces. Their Gromov-Hausdorff
distance $d(X, Y)$ is the infimum of $d^{Z'}_H(X', Y')$ over all metric
spaces $Z'$ and all subsets $X'$ and $Y'$ of $Z'$ which are isometric
respectively to $X$ and $Y$.
\end{definition}

Note that the infimum in this definition makes sense: it is always possible
to embed isometrically $X$ and $Y$ in a common metric space (for instance by
putting a suitable distance on the disjoint union of $X$ and $Y$).

Let $\GH$ denote the ``space'' of nonempty compact metric spaces up to
isometry. There is a set-theoretic difficulty here, to which we will come
back in Section~\ref{sec:logic} but that we will ignore for now.

The basic result in the theory is the following theorem, which we have
formalized in the Lean proof assistant as part of the \mathlib library.

\begin{thm}
\label{thm:GH} The Gromov-Hausdorff distance is indeed a distance on $\GH$.
With this distance, $\GH$ is a complete second countable metric space.
\end{thm}
Let us highlight two important points in this theorem that will be relevant
later on.
\begin{itemize}
\item If two spaces are at distance zero, the theorem asserts that they are
    isometric. This is not obvious as the Gromov-Hausdorff distance is
    defined as an infimum. This result follows from the more general fact
    that the Gromov-Hausdorff distance between two spaces $X$ and $Y$ is
    always realized, i.e., the aforementioned infimum is in fact a minimum.
    To prove this, one should \emph{construct} a metric space $Z'$ and two
    isometric embeddings $f : X \to Z'$ and $g : Y \to Z'$ with
    $d^{Z'}_H(f(X), g(Y)) = d(X, Y)$.
\item Given a Cauchy sequence $X_n$ of compact metric spaces (for instance
    a sequence such that $d(X_n, X_{n+1}) \leq 2^{-n}$), the theorem
    asserts that there exists a compact metric space $X_\infty$ such that
    $X_n \to X_\infty$. Again, this statement involves the
    \emph{construction} of the limiting space $X_\infty$.
\end{itemize}

The standard setting for discussing convergence of random objects in
probability theory is that of complete second countable metric spaces
(see~\cite{billinsgley:convergence}). Thanks to Theorem~\ref{thm:GH}, this
means that a theory of convergence of random compact metric spaces can be set
up, and indeed it has become ubiquitous in modern probability theory. Let us
just mention Aldous' continuous random tree~\cite{aldous_CRT}, which
informally speaking is a random compact metric space which is almost-surely a
(real) tree, but which formally is given by a probability measure on the
space $\GH$. It roughly plays for random metric spaces the same universal
role as Brownian motion does for random walks. The above framework makes it
possible to say rigorously that a family of random metric spaces converges in
distribution to the continuous random tree. This notion shows up in the
author's mathematical research, and is his original motivation to formalize
the Gromov-Hausdorff space in a proof assistant, as a step in his
(unrealistic) program to formalize his own research results.

\section{Formalization overview}

\label{sec:formal}

Formalizing the Gromov-Hausdorff space is an interesting task because of the
unusual feature that it is a ``space'' of equivalence classes of spaces (with
quotation marks around the first space because of set-theoretic issues). A
usable formalization in a mathematics library should retain the following
properties:
\begin{enumerate}
\item It should interact well with preexisting topological concepts. In
    other words, if there is a standard notion of compact topological space
    $X$ in the library, then one should be able to talk of the
    Gromov-Hausdorff distance between such spaces $X$ and $Y$, and not
    between new gadgets that would have been defined specifically in view
    of this formalization.
\item The resulting Gromov-Hausdorff space should also be a topological
    space in the standard sense of the library.
\item One should be able to define a function mapping a nonempty compact
    metric space (in the usual sense) to an element of the Gromov-Hausdorff
    space.
\end{enumerate}

These constraints seem hard to satisfy in simple type theory, where it is not
possible to define a function whose arguments are types (the compact space
$X$) and whose images are elements of another type (the Gromov-Hausdorff
space). On the other hand, they should not be a problem for a framework based
on dependent type theory, with an expressive enough mathematical library. Our
formalization is done using the Lean theorem prover, based on a version of
the calculus of inductive construction, in the framework of the \mathlib
library. It satisfies the above requirements. The main results are available
in the \mathlib
file~\href{https://github.com/leanprover-community/mathlib/blob/master/src/topology/metric_space/gromov_hausdorff.lean}{\texttt{topology/metric\_space/gromov\_hausdorff.lean}}.

Let us give the form of the interface, i.e., the main definitions and
statements, leaving implementation or proof details in $...$ blocks, before
getting to more details.

\begin{lstlisting}
definition GH_space : Type := ...

instance : metric_space GH_space := ...

instance : second_countable_topology GH_space := ...

instance : complete_space GH_space := ...

/-- Mapping a nonempty compact metric space to its equivalence class in `GH_space`. -/
definition to_GH_space (X : Type u) [metric_space X] [compact_space X] [nonempty X] : GH_space := ...

/-- Two nonempty compact spaces have the same image in `GH_space` if and only if they are isometric. -/
theorem to_GH_space_eq_to_GH_space_iff_isometric {X : Type u} [metric_space X] [compact_space X]
  [nonempty X] {Y : Type v} [metric_space Y] [compact_space Y] [nonempty Y] :
  to_GH_space X = to_GH_space Y ↔ nonempty (X ≃ᵢ Y) := ...

/-- The Gromov-Hausdorff distance between two spaces `X` and `Y` can be realized by
isometric embeddings into ` ℓ_infty_ℝ`. -/
theorem GH_dist_eq_Hausdorff_dist (X : Type u) [metric_space X] [compact_space X] [nonempty X]
  (Y : Type v) [metric_space Y] [compact_space Y] [nonempty Y] :
  ∃ Φ : X → ℓ_infty_ℝ, ∃ Ψ : Y → ℓ_infty_ℝ, isometry Φ ∧ isometry Ψ ∧
  GH_dist X Y = Hausdorff_dist (range Φ) (range Ψ) := ...

\end{lstlisting}

In this snippet, \verb+GH_space+ is a type formalizing the Gromov-Hausdorff
space, i.e., the space of nonempty compact metric spaces up to isometry. It
is endowed with a distance (the Gromov-Hausdorff distance) which turns it
into a metric space, in the metric space instance. This metric space turns
out to be second-countable and complete. The interface with concrete nonempty
compact metric spaces is made through the function \verb+to_GH_space+,
associating to any nonempty compact metric space $X$ its equivalence class in
the Gromov-Hausdorff space. This definition is used in the form
\verb+to_GH_space X+: the assumptions of the form \verb+[...]+ in this
definition are \emph{typeclass assumptions} on $X$, registering that it is a
nonempty compact metric space, and filled in automatically by the system when
seeing an expression of the form \verb+to_GH_space X+. The system raises an
error if it can not deduce an instance for these from the context.

The relationship between concrete nonempty compact metric spaces and abstract
points in the Gromov-Hausdorff space is illustrated with two theorems:
\begin{itemize}
\item \verb+to_GH_space_eq_to_GH_space_iff_isometric+ asserts that two
    spaces have the same image in the Gromov-Hausdorff space if and only if
    they are isometric;
\item \verb+GH_dist_eq_Hausdorff_dist+ says that the Gromov-Hausdorff
    distance between two nonempty compact metric spaces is realized, i.e.,
    one can embed them isometrically in a common metric space so that the
    Hausdorff distances between their images is exactly their
    Gromov-Hausdorff distance. The theorem is a little bit stronger,
    because it says that one can use as a common embedding space the metric
    space $\ell^\infty(\R)$ of bounded real sequences, whatever the compact
    metric spaces $X$ and $Y$. (See Section~\ref{sec:logic} for more on
    this).
\end{itemize}
Note that the former theorem is an easy consequence of the latter: if $X$ and
$Y$ have zero Gromov-Hausdorff distance, then their images under $\Phi$ and
$\Psi$ given by the second theorem are at zero Hausdorff distance, hence they
coincide, and it follows that $\Psi^{-1} \circ \Phi$ is an isometry between
$X$ and $Y$.

In the next three sections, we will give more details on three salient points
of the formalization.

\section{Formal definition of the Gromov-Hausdorff space}
\label{sec:logic}

Until now, we have described the Gromov-Hausdorff space as the ``space'' of
equivalence classes of compact metric spaces up to isometry. There is a
problem here: we are quantifying over objects which are not constrained to
belong to a given set. This kind of construction is not allowed in set
theory, as it leads to Russell-like paradoxes: nonempty compact metric spaces
form a class, not a set.

The type theory implemented by Lean makes it possible to circumvent this
issue, thanks to the notion of universe level (already dating back to
Russell). Informally speaking, any class at universe level $u$ becomes an
object one can manipulate at level $u+1$, where $u$ range over $\N$. Thus,
one can define the type of all nonempty compact metric spaces in universe
level $u$, as a well-defined type in universe level $u+1$. Denote it with
$K_u$. One can then define an equivalence relation $\sim$ on $K_u$, saying
that two spaces are isometric, and construct a Gromov-Hausdorff space $\GH_u$
as $K_u/\sim$.

This definition has two drawbacks. First, it depends on the universe level
$u$: one does not get one single Gromov-Hausdorff space, but infinitely many
of them. This makes it more complicated to discuss the Gromov-Hausdorff
distance of two nonempty compact spaces if they come from different
universes. The second issue is that using universe levels for a construction
is often not satisfactory to mathematicians: it means getting out of the
standard ZFC framework by adding inaccessible cardinals axioms.

\medskip

It turns out that the set theoretic issue that there is no set of all compact
metric spaces is not a real issue. Indeed, the cardinality of a compact
metric space is at most the cardinality of the continuum, which means that
the number of non-isometric compact metric spaces is also controlled. For the
formalization, this means that we can define one single Gromov-Hausdorff
space, in the first universe \texttt{Type $0$} (also called simply
\texttt{Type}), the universe in which most natural objects such as $\N$ or
$\R$ live. However, we can not just dismiss the issue as irrelevant, as is
done in most textbooks: we have to make a sensible design choice.

We use the following classical proposition.
\begin{prop}
\label{prop:Kuratoswki} Consider a compact metric space $X$. There exists an
isometric embedding of $X$ into the space $\ell^\infty(\R)$ of bounded real
sequences with its distance coming from the sup norm.
\end{prop}
\begin{proof}
Let $x_n$ be a dense sequence in $X$. To a point $x\in X$, associate the
sequence $n \mapsto d(x, x_n)$. It is easy to check that this defines an
isometric embedding of $X$ into $\ell^\infty(\R)$.
\end{proof}

This embedding is called the Kuratowski embedding. From this proposition, it
follows that all compact metric spaces have isometric representatives as
subsets of $\ell^\infty(\R)$. Therefore, we may \emph{define} the
Gromov-Hausdorff space as the space of all nonempty compact subsets of
$\ell^\infty(\R)$, up to isometry (taking advantage of the built-in quotient
construction of Lean). This is an element of \texttt{Type~$0$} as announced,
as all objects in this definition live in \texttt{Type~$0$}.

The map \verb+to_GH_space+, assigning to an arbitrary nonempty compact metric
space the corresponding point in \verb+GH_space+, is then obtained by taking
an isometric image of $X$ in $\ell^\infty(\R)$ thanks to
Proposition~\ref{prop:Kuratoswki}, and then descending to the quotient
\verb+GH_space+.

\section{The Gromov-Hausdorff distance is realized}
\label{sec:realized}

Consider two nonempty compact metric spaces $X$ and $Y$. A key point to show
that the Gromov-Hausdorff distance is a distance is to show that there exist
a metric space $Z'$ and two isometric copies $X'$ of $X$ and $Y'$ of $Y$
inside $Z'$ such that the Hausdorff distance $d^{Z'}_H(X', Y')$ is equal to
the Gromov-Hausdorff distance $d(X, Y)$ of $X$ and $Y$. One has always $d(X,
Y) \leq d^{Z'}_H(X', Y')$, and the goal is to construct suitable $Z'$ and
$X', Y'$ such that this inequality becomes an equality.

One can always find a sequence of spaces $Z'_n$ and isometric embeddings
$\Phi_n : X \to Z'_n$ and $\Psi_n : Y \to Z'_n$ such that
$d_H^{Z'_n}(\Phi_n(X), \Psi_n(Y))$ converges to $d(X, Y)$, by definition of
an infimum. The difficulty is that the spaces $Z'_n$ are unrelated to each
other, so making things converge by extracting subsequences has no obvious
meaning.

The key idea is to forget completely $Z'_n$, and only remember its distance.
Define a map $\Theta_n$ from the disjoint union $X \sqcup Y$ to $Z'_n$, equal
to $\Phi_n$ on $X$ and to $\Psi_n$ on $Y$. Define a function $d_n$ on $(X
\sqcup Y)^2$ by $d_n(a, b) = d(\Theta_n(a), \Theta_n(b))$, where the distance
on the right hand side is the distance in $Z'_n$. This is almost a distance
on $X \sqcup Y$, coinciding with the original distances on $X$ and on $Y$,
except that it does not satisfy in general $d_n(a,b)= 0 \Rightarrow a=b$
since different points in $X$ and $Y$ may be mapped to the same point in
$Z'_n$.

We claim that $d_n$ has a subsequence which converges uniformly to a function
$d_\infty$ (which is also almost a distance in the previous sense). Define a
space $Z$ to be the quotient of $X\sqcup Y$ identifying two points $a$ and
$b$ when $d_\infty(a, b)=0$. This is a metric space, in which $X$ and $Y$
embed isometrically and realizing the Gromov-Hausdorff distance by
construction.

It remains to check the claim. This is a consequence of the classical
Arzela-Ascoli theorem:
\begin{thm}
Let $f_n : X \to Z$ be a sequence of bounded continuous functions on a
compact space $X$, with range included in a compact subset of the metric
space $Z$. Assume that the functions $f_n$ are equicontinuous: for every $x
\in X$ and every $\epsilon>0$, there exists a neighborhood $U$ of $x$ such
that for every $n\in \N$ and every $y\in U$, one has $d(f_n(y), f_n(x))\leq
\epsilon$. Then $f_n$ admits a uniformly converging subsequence.
\end{thm}

Indeed, one checks readily that the family of functions $d_n$ on the compact
space $(X \sqcup Y)^2$ is equicontinuous (it is even uniformly
Lipschitz-continuous). Unfortunately,the  Arzela-Ascoli theorem was not
available in \mathlib at the time of the formalization of the
Gromov-Hausdorff space (and neither was the notion of uniform convergence!).
An important part of this formalization was therefore devoted to all these
prerequisites, including the definition and study of the Banach space of
bounded continuous functions on a topological space.

It is worth pointing out that, since then, these results have been put to
good use in completely different directions in \mathlib (for instance to
formalize the Stone-Weierstrass theorem, asserting that an algebra of
continuous functions separating points on a compact space is dense in the
space of continuous functions). This is an important point of the
formalization of fairly specialized concepts such as the Gromov-Hausdorff
space: it is a way to notice general-purpose gaps in the library and to fill
them. The \mathlib philosophy is that these gaps should not be filled just in
the minimal way needed to prove the target theorem, but in the maximal
possible generality to make it suitable for further uses in different
directions. For instance, during the formalization of the Gromov-Hausdorff
distance, the notion of uniform convergence has been defined in the maximal
generality of uniform spaces, even if we only needed the case of metric
spaces for this specific application.

\section{Completeness of the Gromov-Hausdorff space}
\label{sec:complete}

To prove the completeness of the Gromov-Hausdorff space, we will need to glue
metric spaces along isometric subspaces, as follows. Assume that $Y$ and $Z$
are two metric spaces, and that another metric space $X$ admits two isometric
embeddings $\Phi : X \to Y$ and $\Psi : X \to Z$. Then one can form a new
space by identifying the two $X$ subsets in $Y$ and $Z$, and this new space
is naturally a metric space, containing isometric copies of both $Y$ and $Z$.

Let us now explain why the Gromov-Hausdorff space is complete. It is enough
to show that a sequence of compact spaces satisfying $d(X_n, X_{n+1}) \le
2^{-n}$ converges. The difficulty is that we need to construct in some way
the limiting metric space. The idea is to embed simultaneously all the $X_n$
in a common metric space $Z_\infty$, with controlled mutual Hausdorff
distances, and use the fact the space of subsets of $Z_\infty$, with the
Hausdorff distance, is complete, to get the desired limit at a subset of
$Z_\infty$.

There exists for each $n$ a metric space $Y_n$ containing isometric copies of
$X_{n-1}$ and of $X_n$ which are at Hausdorff distance $\le 2^{-n}$, by
Section~\ref{sec:realized}. Let us now define inductively a sequence of
metric spaces $Z_n$ containing isometric copies of $X_0,\dotsc, X_n$, as
follows.
\begin{enumerate}
\item Start with $Z_0 = X_0$.
\item Assume $Z_n$ is defined. It contains an isometric copy of $X_n$. So
    does $Y_{n+1}$. Therefore, we may glue $Z_n$ and $Y_{n+1}$ along their
    respective copies of $X_n$, to obtain the new space $Z_{n+1}$. It
    contains a copy of $X_{n+1}$ (the one contained in $Y_{n+1}$) and
    copies of $X_0,\dotsc, X_n$ (the ones contained in $Z_n$).
\end{enumerate}
Define a suitable limit $Z_\infty$ of the increasing family $Z_n$ (formally,
an inductive limit). It is a metric space, containing for each $k$ a copy
$X'_k$ of $X_k$. By construction, the Hausdorff distance
$d^{Z_\infty}_H(X'_k, X'_{k+1})$ is $\le 2^{-k}$. Since, on a given complete
metric space, the space of its compact subsets is a complete space for the
Hausdorff distance, it follows that $X'_k$ converges, to a compact subset
$X'_\infty$ of $Z_\infty$. Then, in the Gromov-Hausdorff space, $X_n$
converges to the class of $X'_\infty$.

There is an interesting feature in the formalization of this proof, in the
inductive definition of the space $Z_n$, highlighting several shortcomings of
Lean 3. One should define simultaneously the space $Z_n$, but also a metric
space structure on it, and an isometric embedding of $X_n$ in $Z_n$ (which
only makes sense given the metric space structure). And the next step of the
construction will take advantage of all these data to proceed. The most
natural formalization would be by several mutually inductive definitions, but
Lean 3 has weaknesses in this area. Instead, we used one single structure
containing all these data, and one big induction to define the structure at
step $n+1$ from the structure at step $n$. Another issue is that the Lean 3
equation compiler generates a definition in terms of bounded recursion which
is not easy to use. We use instead a direct definition in terms of the
recursor for natural numbers. Here is the full inductive definition we use.

\begin{lstlisting}
variables (X : ℕ → Type) [∀ n, metric_space (X n)] [∀ n, compact_space (X n)] [∀ n, nonempty (X n)]

/-- Auxiliary structure used to glue metric spaces below, recording an isometric embedding
of a type `A` in another metric space. -/
structure aux_gluing_struct (A : Type) [metric_space A] : Type 1 :=
(space  : Type)
(metric : metric_space space)
(embed  : A → space)
(isom   : isometry embed)

/-- Auxiliary sequence of metric spaces, containing copies of `X 0`, ..., `X n`, where each
`X i` is glued to `X (i+1)` in an optimal way. The space at step `n+1` is obtained from the space
at step `n` by adding `X (n+1)`, glued in an optimal way to the `X n` already sitting there. -/
def aux_gluing (n : ℕ) : aux_gluing_struct (X n) := nat.rec_on n
  { space  := X 0,
    metric := by apply_instance,
    embed  := id,
    isom   := λ x y, rfl }
(λ n Z, by letI : metric_space Z.space := Z.metric; exact
  { space  := glue_space Z.isom (isometry_optimal_GH_injl (X n) (X (n+1))),
    metric := by apply_instance,
    embed  := (to_glue_r Z.isom (isometry_optimal_GH_injl (X n) (X (n+1))))
                        ∘ (optimal_GH_injr (X n) (X (n+1))),
    isom   := (to_glue_r_isometry _ _).comp (isometry_optimal_GH_injr (X n) (X (n+1))) })
\end{lstlisting}

We start from a context in which a sequence of nonempty compact metric spaces
$X_n$ is given, which we want to glue together. The structure
\verb+aux_gluing_struct A+ records a metric space containing an isometric
copy of a metric space $A$. The definition \verb+aux_gluing n+ constructs
inductively over $n$ a metric space containing an isometric copy of $X_n$
(and also of all the previous ones, by design, but we only register the last
one for the inductive construction). For $n=0$, it is just $X_0$. At the
$(n+1)$-th step, it glues two spaces containing an isometric copy of $X_n$
along $X_n$ as explained in Section~\ref{sec:realized}: on the one hand the
space constructed at the previous step; on the other hand a space in which
the Gromov-Hausdorff distance between $X_n$ and $X_{n+1}$ is realized.

The reader may note the line \verb+letI : metric_space Z.space := Z.metric+
at the beginning of the inductive step in the construction. By induction, the
space \verb+Z.space+ constructed at step $n$ has a metric space structure,
called \verb+Z.metric+. However, this metric space structure is not yet
available to typeclass inference: there is a caching mechanism underneath
(which is very important performancewise as typeclass inference is quite
costly), so new instances need to be declared explicitly just like here. Once
this preliminary incantation has been done, the system knows about the metric
space structure on \verb+Z.space+ and is happy with the statement that a map
to \verb+Z.space+ is an isometry, for instance. This is needed for the next
step of the construction to go through.

Once this inductive definition has been set up properly (together with enough
properties of the gluing of metric spaces, and of inductive limits of metric
spaces), the rest can be formalized without any specific difficulty.

\bibliography{biblio}

\newcommand{\etalchar}[1]{$^{#1}$}
\def\cprime{$'$}
\providecommand{\bysame}{\leavevmode\hbox to3em{\hrulefill}\thinspace}
\providecommand{\MR}{\relax\ifhmode\unskip\space\fi MR }
\providecommand{\MRhref}[2]{%
  \href{http://www.ams.org/mathscinet-getitem?mr=#1}{#2}
}
\providecommand{\href}[2]{#2}
\begin{thebibliography}{dMKA{\etalchar{+}}15}

\bibitem[{Ald}91]{aldous_CRT}
David {Aldous}, \emph{{The continuum random tree. I}}, {Ann. Probab.}
  \textbf{19} (1991), no.~1, 1--28 (English).

\bibitem[BBI01]{burago_burago_ivanov}
D.~{Burago}, Yu. {Burago}, and S.~{Ivanov}, \emph{{A course in metric
  geometry}}, vol.~33, Providence, RI: American Mathematical Society (AMS),
  2001 (English).

\bibitem[Bil99]{billinsgley:convergence}
Patrick Billingsley, \emph{Convergence of probability measures}, second ed.,
  Wiley Series in Probability and Statistics: Probability and Statistics, John
  Wiley \& Sons Inc., New York, 1999, A Wiley-Interscience Publication.
  \MR{MR1700749}

\bibitem[dMKA{\etalchar{+}}15]{demoura_lean}
Leonardo de~Moura, Soonho Kong, Jeremy Avigad, Floris van Doorn, and Jakob von
  Raumer, \emph{The lean theorem prover (system description)}, Automated
  deduction---{CADE} 25, Lecture Notes in Comput. Sci., vol. 9195, Springer,
  Cham, 2015, pp.~378--388. \MR{3446905}

\bibitem[{Gro}81]{gromov_hausdorff}
Mikhael {Gromov}, \emph{{Structures m\'etriques pour les vari\'et\'es
  riemanniennes. R\'edig\'e par J. Lafontaine et P. Pansu}}, 1981 (French).

\bibitem[mc20]{mathlib}
The mathlib community, \emph{The {L}ean mathematical library}, Proceedings of
  the 9th {ACM} {SIGPLAN} International Conference on Certified Programs and
  Proofs, {CPP} 2020, 2020, pp.~367--381.

\bibitem[{Pet}16]{petersen_book}
Peter {Petersen}, \emph{{Riemannian geometry. 3rd edition}}, 3rd edition ed.,
  vol. 171, Cham: Springer, 2016 (English).

\end{thebibliography}
\bibliographystyle{amsalpha}

\end{document}